\documentclass[prb,twocolumn,floatfix,nofootinbib]{revtex4-1}
\usepackage{amsmath,amssymb,amsthm}
\usepackage{graphicx} 
\usepackage{epstopdf}
\usepackage[colorlinks,allcolors=blue]{hyperref}
\newtheorem{lemma}{Lemma}
\def\bm{\boldsymbol}
\def\rmd{\mathrm{d}}
\newcommand{\unit}{1\!\!1}
\makeatletter
\newcommand*{\defeq}{\mathrel{\rlap{%
                     \raisebox{0.3ex}{$\m@th\cdot$}}%
                     \raisebox{-0.3ex}{$\m@th\cdot$}}%
                     =}
\makeatother

\def\x{\end{document}}

\begin{document}

\title{Exact energy quantization condition for single Dirac particle in one-
dimensional (scalar) potential well}

\author{Siddhant Das}
\affiliation{Elite master program `Theoretical and Mathematical Physics'\\
Arnold Sommerfeld Center, Ludwig-Maximilians-Universit\"{a}t M\"{u}nchen}

\begin{abstract}
We present an exact quantization condition for the time independent solutions 
(energy eigenstates) of the one-dimensional Dirac equation with a scalar 
potential well characterized by only two `effective' turning points (defined 
by the roots of $V(x)+mc^2=\pm E$) for a given energy $E$ and satisfying $mc^
2+\min V(x)\geq 0$. This result generalizes the previously known 
non-relativistic quantization formula and preserves many physically desirable 
symmetries, besides attaining the correct non-relativistic limit. Numerical 
calculations demonstrate the utility of the formula for computing accurate 
energy eigenvalues.
\end{abstract}
\maketitle
\section{Introduction}\label{sec1}

In this paper we present a quantization formula for the one-dimensional Dirac 
equation using the analytical transfer matrix (ATM) method, thereby 
extending the non-relativistic analogue for the Schr\"{o}dinger equation 
found by Cao et al. \cite{1,2,3}. We focus on the relativistic energy levels 
offered by a simple confining scalar potential well with two turning points.

The high accuracy of the eigenvalues computable from this formula together 
with other applications like ground state reconstruction should prove to be 
useful in many applications of the Dirac equation, especially in the context 
of solid state physics \cite{7,rev}.

Our result can be thought of as a completion of well known semi-classical 
quantization formulae (like Bohr--Sommerfeld and WKB) in the following 
sense. The semi-classical quantization formulae provide reliable estimates 
of the energy eigenvalues only in the limit of large quantum numbers, while 
the ATM-quantization formula gives exact eigenvalues for all quantum 
numbers. As a result, Cao et al. additionally characterize this formula as 
an `exact quantization formula.'

In anticipation of the derivation given below, we would like to pursue this 
comparison a little further. The reader will recall that, in the WKB 
quantization formula, for instance, the boundary conditions at the turning 
points when applied to a suitably chosen approximate ansatz (for the actual 
wave function), lead to quantization of energy. However, the ATM method is 
based on transfer matrices which, in the appropriate limit `recovers' the 
exact wave function $\Psi$ and the (corresponding) energy eigenvalue $E$. 
Structurally, the exact quantization formula contains the usual WKB term 
with a non-trivial correction contributed by the so-called `scattered sub-
waves', which is of great significance (discussed in Section~\ref{sec2}).

Particularly, for the Dirac equation the negative energy (antiparticle) 
solutions make the energy spectrum unbounded from below, which for general 
potentials is difficult to account for using the ATM method. Also, the 
generalization of this method for even the Schr\"{o}dinger equation to 
potentials having more than two classical turning points ($x$ for which $V(x
)=E$) is not clear at present. This difficulty translates to our inability 
to obtain the quantization condition for the Dirac problem for potential 
wells that either (1) give more than two classical turning points, or (2) 
satisfy $mc^2+\min V(x)<0$, or both. Owing to these limitations, we narrow 
our focus to potential wells with only two classical turning points and 
satisfying  $mc^2+\min V(x)\geq 0$. For this restricted class of potential 
wells we can apply the ATM method successfully.

Following up with the derivation in Section~\ref{sec2} we point to some 
desirable symmetries of the quantization formula in Section~\ref{sec3} and 
discuss its non-relativistic limit. Section~\ref{sec4} is devoted to 
numerical results and applications. We conclude in Section~\ref{sec5} 
outlining prospects of further study.


\section{Formulation}\label{sec2}

Consider the one-particle Dirac Hamiltonian $H=c\bm{\alpha}p+\bm{\beta}(mc^2+V
)$, where $p=-i\hbar\rmd /\rmd x$ is the momentum operator and $m$ ($c$) is 
the rest mass of the particle (speed of light). We choose to represent the 
Dirac matrix $\bm{\alpha}(\bm{\beta})$ by the Pauli matrix $\bm{\sigma}_y (
\bm{\sigma}_z)$, which has the advantage that the two-component wave 
function $\Psi=
\begin{pmatrix}\psi_1&\psi_2
\end{pmatrix}^\top$ can be chosen 
to be real \cite{4,5}. The time independent Dirac equation prescribes the 
eigenvalue problem $H\Psi=E\Psi$, where $E$ is the energy of the particle. 
Further, the property $H\Psi'=-E\Psi'\Rightarrow\Psi'=\bm{\sigma}_x\Psi$ 
shows that the positive and negative solutions occur in pairs. Hence, it 
suffices to consider positive energy solutions alone.

The (reduced) Compton wavelength $\hbar/mc$ and rest mass energy $mc^2$ 
provide natural length and energy scales in this problem. Thus, expressing 
every quantity in dimensionless form $x\rightarrow\hbar/mc\ \xi$,   $E
\rightarrow mc^2\varepsilon$, $V\rightarrow mc^2\vartheta$, the Dirac 
equation translates into the coupled first order system of equations
\begin{gather}
\dot{\psi}_1=(1+\vartheta(\xi)+\varepsilon)\psi_2\label{eq:1}\\
\dot{\psi}_2=(1+\vartheta(\xi)-\varepsilon)\psi_1\label{eq:2}
\end{gather}
the overhead dot denoting differentiation w.r.t. $\xi$. We identify solutions 
of the equations $1+\vartheta(\xi)=\pm\varepsilon$ as effective turning 
points. For the schematic potential well depicted in Fig. \ref{fig1}, two 
turning points $\xi_{L,R}(\xi_L<\xi_R)$ are obtained for a typical energy $
\varepsilon$. Note that more than two effective turning points would result 
for $\vert\varepsilon\vert<-(1+\min\vartheta)$ unless the restrictions noted 
in Section~\ref{sec1} are enforced on $V$.
\begin{figure}
\begin{center}
\includegraphics[width=0.9\columnwidth]{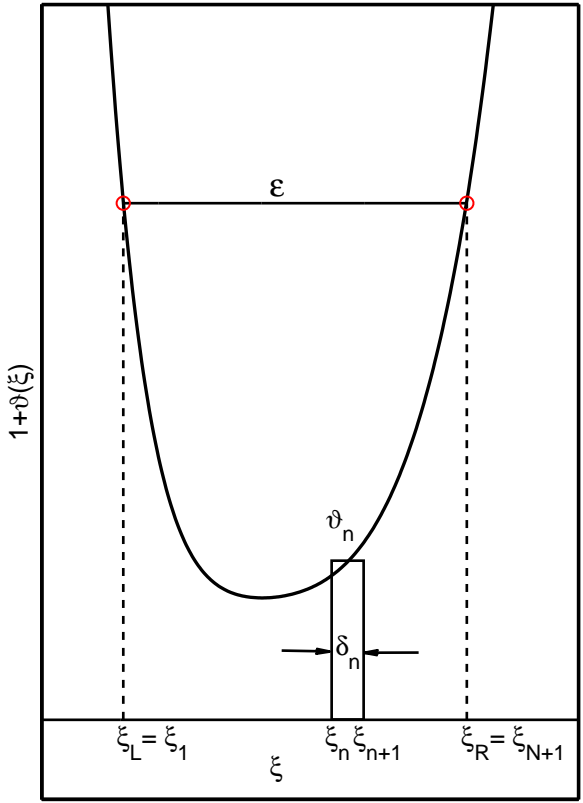}
\end{center}
\caption{The (scalar) potential well $\vartheta(\xi)$, showing effective 
turning points $\xi_{L,R}$ for a typical energy $\varepsilon$.}\label{fig1}
\end{figure}

Next, we partition the interval between the two turning points into $N$ 
segments $\xi_L=\xi_1<\xi_2\dotsm<\xi_N<\xi_{N+1}=\xi_R$. Denoting the width 
of the $n$\textsuperscript{th} segment by $\delta_n$, we choose an arbitrary 
point (a tag) $\xi'_n$ within each segment and replace the scalar potential 
by a piecewise-constant approximation for which the potential in the $n$
\textsuperscript{th} segment is given by $\vartheta_n=\vartheta(\xi'_n)$. 
Consequently, on this segment the system of Eqs. \eqref{eq:1} and 
\eqref{eq:2} admits a general solution of the form
\begin{equation}
\label{eq:3}
\Psi^n = A_n
\begin{pmatrix}1\\{i\lambda_n}
\end{pmatrix}e^{i\kappa_n(\xi-\xi_n)
}+B_n
\begin{pmatrix}1\\{-i\lambda_n}\end{pmatrix}e^{-i\kappa_n(\xi-\xi_n)},
\end{equation}
where $A_n$ and $B_n$ are arbitrary coefficients, and
\begin{align*}
\kappa_n=\sqrt{\varepsilon^2-(1+\vartheta_n)^2},\ \lambda_n=\sqrt{\frac{
\varepsilon-\vartheta_n-1}{\varepsilon+\vartheta_n+1}}.
\end{align*}
Note that in the region of interest $\kappa_n$, $\lambda_n$ are real, hence 
the wave function components $\psi_{1,2}^n$ are oscillatory. In favor of a 
simpler notation, we drop the superscript $n$ from the wave function and 
infer the segment label from the context.

Ensuring the continuity of $\Psi^n$ at the ends of the segment leads to the (
transfer) matrix equation
\begin{equation}\label{eq:4}
\begin{pmatrix}\psi_1(\xi_n)\\\psi_2(\xi_n)\end{pmatrix}= \begin{pmatrix}
{\cos(\kappa_n\delta_n)}&{-\frac{\sin(\kappa_n\delta_n)}{\lambda_n}}\\{\lambda
_n\sin(\kappa_n\delta_n)}&{\cos(\kappa_n\delta_n)}\end{pmatrix}
\begin{pmatrix}\psi_1(\xi_{n+1})\\\psi_2(\xi_{n+1})\end{pmatrix}
\end{equation}
Since we allow only two turning points at this stage, $\lambda_n$ is nonzero 
in each segment. Consequently, the transfer matrices are always well 
defined. Now, left multiplying Eq.~\eqref{eq:4} by $\begin{pmatrix}-\psi_2(
\xi_n)&\psi_1(\xi_n)\end{pmatrix}$ and dividing the resulting equation by $
\psi_1(\xi_n)\psi_1(\xi_{n+1})$, we arrive at
\begin{equation*}
\begin{pmatrix}P_n&1\end{pmatrix}\begin{pmatrix}
{\cos(\kappa_n\delta_n)}&{-\frac{\sin(\kappa_n\delta_n)}{\lambda_n}}\\{\lambda
_n\sin(\kappa_n\delta_n)}&{\cos(\kappa_n\delta_n)}\end{pmatrix}
\begin{pmatrix}1\\-P_{n+1}\end{pmatrix}=0\\
\end{equation*}
where $P_j=-\psi_2(\xi_j)/\psi_1(\xi_j)$. Simplifying this matrix equation 
gives the recurrence formula
\begin{align}\label{eq:5}
&\frac{P_n}{\lambda_n}=\tan\bigg(\tan^{-1}\bigg(\frac{P_{n+1}}{\lambda_n}\bigg
)-\kappa_n\delta_n\bigg),\\
&\Rightarrow\tan^{-1}\bigg(\frac{P_n}{\lambda_n}\bigg)-\tan^{-1}\bigg(\frac{P_
{n+1}}{\lambda_n}\bigg)=z\pi-\kappa_n\delta_n,\nonumber\\
&\hspace{14em} z=0,1,2,\ldots\label{eq:6}
\end{align}
Rearranging Eq.~\eqref{eq:6} and summing over $n$ from 1 to $N$ yields
\begin{multline}
\sum_{n=1}^N\kappa_n\delta_n + \sum_{n=1}^{N-1}\tan^{-1}\bigg(\frac{P_{n+1}}{
\lambda_{n+1}}\bigg)-\tan^{-1}\bigg(\frac{P_{n+1}}{\lambda_n}\bigg)\\
=z\pi
+\tan^{-1}\bigg(\frac{P_{N+1}}{\lambda_N}\bigg)-\tan^{-1}\bigg(\frac{P_1}{
\lambda_1}\bigg)\ z=0,1,2,\ldots \label{eq:7}
\end{multline}
The desired quantization condition emerges as a limit of Eq.~\eqref{eq:7} as $
\delta(\defeq\max\delta_n)\rightarrow 0$. In the limiting event, the 
continuous potential variation $\vartheta(\xi)$ is recovered, with $P_1
\rightarrow P(\xi_L)$ and $P_{N+1}\rightarrow P(\xi_R)$. In 
Appendix~\ref{appA} we prove that (for bound state wave functions) $P(\xi_L)<0<P(\xi_R)$ 
and $|P(\xi_{L,R})|<\infty$. Also, as $\delta\rightarrow 0$, $\lambda_{1,N}
\rightarrow \lambda(\xi_{L,R})=0$. As a result, the `half phase losses' at 
the effective turning points are given by
\begin{align}
\lim_{\delta\rightarrow 0}\tan^{-1}\left(\frac{P_{N+1}}{\lambda_N}\right)&=-
\lim_{\delta\rightarrow 0}\tan^{-1}\left(\frac{P_1}{\lambda_1}\right)=\frac{
\pi}{2},\nonumber\\
\delta&\defeq\max_n\delta_n. \label{eq:8}
\end{align}
Next, we account for the phase contribution of the so-called scattered sub-
waves. Defining the phase contribution for each segment
\begin{align*}
\Delta\phi_n &\defeq\tan^{-1}\left(\frac{P_{n+1}}{\lambda_{n+1}}\right)-\tan^{
-1}\left(\frac{P_{n+1}}{\lambda_n}\right)\\
&=\tan^{-1}\left(\frac{P_{n+1}\left(\lambda_n-\lambda_{n+1}\right)}{P^2_{n+1}+
\lambda_n\lambda_{n+1}}\right)\\
&=-\frac{P_{n+1}\left(\lambda_{n+1}-\lambda_n\right)}{P^2_{n+1}+\lambda_n
\lambda_{n+1}}+\mathcal{O}\left(\left(\lambda_{n+1}-\lambda_n\right)^3\right),
\end{align*}
which results from expanding the inverse tangent in a Taylor series in powers 
of $\left(\lambda_{n+1}-\lambda_n\right)$, we obtain the total phase 
contribution of the scattered sub-waves by
\begin{multline}
\lim_{\delta\rightarrow 0}\sum_{n=1}^{N-1}\Delta\phi_n=-\lim_{
\delta\rightarrow 0}\sum_{n=1}^{N-1}\frac{P_{n+1}\left(\lambda_{n+1}-\lambda_
n\right)}{P^2_{n+1}+\lambda_n\lambda_{n+1}}+\\
\lim_{\delta\rightarrow 0}\sum_{n=1}^{N-1}\mathcal{O}\left(\left(\lambda_{n+1}
-\lambda_n\right)^3\right)=-\int_{\xi_L}^{\xi_R}\frac{P\dot{\lambda}}{P^2+
\lambda^2}\rmd \xi \label{eq:9}
\end{multline}
Thus, in the limit $\delta\rightarrow 0$, Eq.~\eqref{eq:7} takes the form
\begin{equation}\label{eq:10}
\int_{\xi_L}^{\xi_R}\kappa-\left(\frac{P\dot{\lambda}}{P^2+\lambda^2}\right)
\rmd \xi=(z+1)\pi,\quad z=0,1,2...
\end{equation}
We can phrase this result differently by substituting
\begin{align}\label{eq:11}
\lambda^2=\frac{\varepsilon-\vartheta(\xi)-1}{\varepsilon+\vartheta(\xi)+1},\ 
P=-\frac{\psi_2}{\psi_1},
\end{align}
and using Eqs. \eqref{eq:1} and \eqref{eq:2}, obtaining
\begin{align}
&\frac{P}{P^2+\lambda^2}\dot{\lambda}=\frac{\psi_1\psi_2(1+\vartheta(\xi)+
\varepsilon)}{(1+\vartheta(\xi)+\varepsilon)\psi_2^2-(1+\vartheta(\xi)-
\varepsilon)\psi_1^2}\nonumber\\
&\times\frac{-\varepsilon\dot{\vartheta}(\xi)}{\lambda(1+\vartheta(\xi)+
\varepsilon)^2}
=\frac{\varepsilon\dot{\vartheta}(\xi)}{\kappa}\left(\frac{\psi_1\psi_2}{\dot{
\psi_1}\psi_2-\dot{\psi_2}\psi_1}\right)\nonumber\\
&=\frac{\varepsilon\dot{\vartheta}(\xi)}{2\kappa}\left(\frac{\Psi^\dagger\bm{
\sigma}_x\Psi}{\Psi^\dagger\left(-i\bm{\sigma}_y\rmd /\rmd \xi\right)\Psi}
\right) \label{eq:12}\\
&=-\frac{\varepsilon\dot{\kappa}}{4i\sqrt{\varepsilon^2-\kappa^2}}\left(\frac{
\Psi^\dagger[\bm{\sigma}_y,\bm{\sigma}_z]\Psi}{\Psi^\dagger\left(-i\bm{\sigma
}_y\rmd /\rmd \xi\right)\Psi}\right)\nonumber\\
&\hspace*{11em}\because\ 2i\bm{\sigma}_x=[\bm{\sigma}_y,\bm{\sigma}_z]. \label
{eq:13}
\end{align}
Using this result and replacing the Pauli matrices with the Dirac matrices 
gives a representation independent form (Ref. Section~\ref{sec3}) of Eq. 
\eqref{eq:10}. Further, restoring the dimensions of the physical quantities 
and using $n$ (instead of $z$) for the quantum number yields
\begin{multline}
\int_{x_L}^{x_R}\frac{K}{\hbar c} + \bigg(\frac{E}{4\sqrt{K^2-E^2}}\frac{\rmd 
K}{\rmd x}\bigg)\bigg(\frac{\Psi^\dagger[\bm{\alpha},\bm{\beta}]\Psi}{\Psi^
\dagger(c\bm{\alpha}p)\Psi}\bigg)\rmd x\\
=(n+1)\pi,\\  n = 0,1,2,\ldots \label{eq:14}
\end{multline}
where $K=\sqrt{E^2-(mc^2+V(x))^2}$, and the effective turning points satisfy $
V(x_{L,R})+mc^2=\pm E$.

\section{Discussion}\label{sec3}
We begin by describing two symmetries of Eq.~\eqref{eq:14} that are 
physically desirable. 

\subsection{Representation independence}\label{sec3.1} 
In Section~\ref{sec2} we chose a convenient representation of the Dirac 
matrices $\bm{\alpha}(\bm{\beta})=\bm{\sigma}_y(\bm{\sigma}_z)$. However, 
the quantized energy levels are a property of the potential $V(x)$, hence 
should be independent of the chosen representation. This feature is already 
built into the quantization condition and can be shown in the following way. 
We know that the Dirac matrices satisfy the algebra: $\bm{\alpha}^2=\bm{\beta
}^2=\unit$ and $\{\bm{\alpha},\bm{\beta}\}=0$. Given any other 
representation of the Dirac matrices (denoted $\bm{\alpha}',\bm{\beta}'$) 
satisfying the same algebra, (a generalization of) Pauli's fundamental 
theorem \cite{6} asserts that there exists a unique (invertible) matrix $\bm{
S}$ (up to a multiplicative complex constant) such that 
\begin{equation}\label{eq:15}
\bm{\gamma}=\bm{S}^{-1}\bm{\gamma}'\bm{S},\qquad \bm{\gamma}=\bm{\alpha},\bm{
\beta}
\end{equation}
A direct substitution shows that this transformation preserves the structure 
of the Dirac equation (for the same $E$) with the wave function transforming 
as $\Psi'=\bm{S}\Psi$. Additionally, as the Dirac matrices are hermitian, we 
must have $\bm{S}^\dagger=\bm{S}^{-1}$ (unitarity of $\bm{S}$). Using the 
transformation Eq.~\eqref{eq:15} in Eq.~\eqref{eq:14} and noting that $[\bm{
\alpha},\bm{\beta}]=\bm{S}^\dagger[\bm{\alpha}',\bm{\beta}']\bm{S}$ gives
\begin{equation*}
\frac{\Psi^\dagger[\bm{\alpha},\bm{\beta}]\Psi}{\Psi^\dagger(c\bm{\alpha}p)
\Psi}\mapsto\frac{\Psi'^\dagger[\bm{\alpha'},\bm{\beta'}]\Psi'}{\Psi'^\dagger
(c\bm{\alpha'}p)\Psi'}
\end{equation*}
thus confirming the representation independence of the obtained quantization 
condition.

\subsection{Symmetry w.r.t. the sign of $\bm{E}$}\label{sec3.2} 
Although we assumed $E$ to be positive in the derivation of Eq.~\eqref{eq:14}, 
the quantization condition should not depend on the sign of the 
eigenvalue. Since $K$ is a function of $E^2$, the apparent asymmetry stems 
from the factor
\begin{equation*}
\varepsilon\frac{\Psi^\dagger\bm{\sigma}_x\Psi}{\Psi^\dagger\left(-i\bm{\sigma
}_y\rmd /\rmd \xi\right)\Psi}
\end{equation*}
in Eq.~\eqref{eq:12}. Using the properties $\bm{\sigma}_x^2=\unit$, $\bm{
\sigma}_x=\bm{\sigma}_x^\dagger,\ \{\bm{\sigma}_x,\bm{\sigma}_y\}=0$ we obtain
\begin{align*}
\varepsilon\frac{\Psi^\dagger\unit\bm{\sigma}_x\Psi}{\Psi^\dagger\unit\left(-i
\bm{\sigma}_y\rmd /\rmd \xi\right)\Psi}&=\frac{\varepsilon(\bm{\sigma}_x\Psi)
^\dagger\bm{\sigma}_x(\bm{\sigma}_x\Psi)}{(\bm{\sigma}_x\Psi)^\dagger\left(-i
\bm{\sigma}_x\bm{\sigma}_y\rmd /\rmd \xi\right)\Psi}
\\&=\frac{(-\varepsilon)(\bm{\sigma}_x\Psi)^\dagger\bm{\sigma}_x(\bm{\sigma}_x
\Psi)}{(\bm{\sigma}_x\Psi)^\dagger\left(-i\bm{\sigma}_y\rmd /\rmd \xi\right)(
\bm{\sigma}_x\Psi)}
\end{align*}
which in the light of our earlier observation $H\Psi=E\Psi \Leftrightarrow H(
\bm{\sigma}_x\Psi)=-E(\bm{\sigma}_x\Psi)$ establishes the expected symmetry. 
As the effective turning points remain unchanged under $E\mapsto-E$, Eq. 
\eqref{eq:14} also holds for the antiparticle solutions.

Finally, we look at the non-relativistic limit of Eq.~\eqref{eq:14} and 
affirm that it reproduces the Schr\"{o}dinger quantization formula found by 
Cao et al. \cite{1,2,3} in this limit. The non-relativistic limit con cerns 
energies $E\approx mc^2$; and is easily demonstrated by the replacements: (1)
 $mc^2+V(x)+E\approx 2mc^2$ and (2) $mc^2+V(x)-E\approx V(x)-E_s$, where $E_
s \defeq E - mc^2$ (the `shifted' energy) is a small quantity (Cf. Sec. 4.4 
of Ref.~[\onlinecite{7}])\footnote{The symbol `$\approx$' implies that the equation under 
consideration resulted from applying prescriptions (1) and (2) to its (
otherwise) relativistic counterpart.}.
Using this prescription, Eq.~\eqref{eq:1} and \eqref{eq:2} (after replacing 
the dimensions) become
\begin{align}
&\frac{\rmd \psi_1}{\rmd x}\approx\frac{2mc}{\hbar}\psi_2\label{eq:16}\\
&\frac{\rmd \psi_2}{\rmd x}\approx\frac{V(x)-E_s}{\hbar c}\psi_1\label{eq:17}
\end{align}
which further decouple to yield the one-dimensional Schr\"{o}dinger equation
\begin{equation}\label{eq:18}
\frac{\rmd ^2\psi_1}{\rmd x^2}\approx\frac{2m}{\hbar^2}\left(V(x)-E_s\right)
\psi_1
\end{equation}
From Eq.~\eqref{eq:11} we arrive at the limiting forms of $\lambda$ and $P$
\begin{align}\label{eq:19}
&\lambda\approx\sqrt{\frac{E_s-V(x)}{2mc^2}}
& P\approx\frac{\hbar}{2mc}P_s
\end{align}
with
\begin{equation}\label{eq:20}
P_s\defeq-\frac{1}{\psi_1}\frac{\rmd \psi_1}{\rmd x}
\end{equation}
where the subscript $s$ is intended for notational homogeneity of later 
equations. Proceeding further, we substitute these results into 
Eq.~\eqref{eq:10} making the following observations
\begin{align*}
\kappa&\approx\sqrt{\frac{2}{mc^2}\left(E_s-V(x)\right)}\\
&=\frac{\hbar}{mc}\sqrt{\frac{2m}{\hbar^2}\left(E_s-V(x)\right)}\defeq\frac{
\hbar}{mc}\kappa_s
\end{align*}
\begin{align*}
\frac{P}{P^2+\lambda^2}\frac{\rmd \lambda}{\rmd \xi}&\approx\frac{P_s}{P_s^2+
\frac{2m}{\hbar^2}\left(E_s-V(x)\right)}\left(2\frac{\rmd \lambda}{\rmd x}
\right)\\
&=\frac{P_s}{P_s^2+\kappa_s^2}\left(-\frac{\rmd V/\rmd x}{\sqrt{2mc^2\left(E_s
-V(x)\right)}}\right)\\
&=\frac{P_s}{P_s^2+\kappa_s^2}\left(\frac{\rmd }{\rmd x}\sqrt{\frac{2}{mc^2}
\left(E_s-V(x)\right)}\right)\\
&=\frac{\hbar}{mc}\frac{P_s}{P_s^2+\kappa_s^2}\frac{\rmd \kappa_s}{\rmd x}
\end{align*}
thus arriving at ($\rmd \xi=mc/\hbar\ \rmd x$)
\begin{equation}\label{eq:21}
\int_{x_{L_s}}^{x_{R_s}}\kappa_s-\frac{\rmd \kappa_s}{\rmd x}\left(\frac{P_s}{
P_s^2+\kappa_s^2}\right)\rmd x\approx(z+1)\pi,\quad z=0,1,2...
\end{equation}
where $x_{L_s}$($x_{R_s}$) denotes the left (right) `classical' turning point 
given by the smaller (larger) of the two roots of the equation $V(x)=E_s$. 
The above equation is the correct non-relativistic limit of Eq.~\eqref{eq:14}, 
which agrees with the well-known quantization formula describing the bound 
states of the Schr\"{o}dinger equation \cite{1,2,3}.

\section{Numerical results}\label{sec4}
We turn now to a discussion of computational aspects of Eq.~\eqref{eq:14} and 
outline some applications of the same. Before proceeding, we point to an 
apparent impediment, namely, the integrand singularities at the turning 
points $x_{L,R}$ arising as a result of
\begin{equation}
\frac{E}{\sqrt{K^2-E^2}}\frac{\rmd K}{\rmd x}=i\frac{E}{K}\frac{\rmd V}{\rmd x
},\quad K(x_{L,R})=0.\label{eq:22}
\end{equation}
Without further qualification, such points might cause the integral to 
diverge. Luckily, this does not happen, as we rigorously show in Appendix 
\ref{appB}. Furthermore, convergence requirements do not restrict the class 
of admissible potentials to which our analysis applies. With that caveat, we 
look at the problem of determining energy eigenvalues of the Dirac equation 
for a given scalar potential $V(x)$.

The collection of one-dimensional potentials for which the Dirac equation is 
closed-form solvable is rather small. Even for the modest (scalar) `simple' 
harmonic oscillator, the wave-functions cannot be obtained in closed 
analytic form. To our knowledge the one-dimensional Woods-Saxon potential 
\cite{8}, the linear confining potential \cite{5,9} and the scalar 
exponential potential \cite{10} have enjoyed closed form solutions so far. 
Thus, an accurate quantization formula becomes particularly useful.

Consider the linear confining potential $V=g|x|,\ g>0$ whose wave functions 
are written in terms of the Hermite function $H_{\nu}(x)$ \cite{hermit}. 
Adopting the conventions $\hbar=c=1$, $\nu=E^2/2g$ given in Ref.~[\onlinecite{9}], we 
wish to incorporate the known wave functions into Eq.~\eqref{eq:14} and 
reproduce the energy eigenvalues, which are otherwise obtained as solutions 
of the transcendental equation $H_{\nu}^2(\alpha)=2\nu H_{\nu-1}^2(\alpha)$, 
where $\alpha=m/\sqrt{g}$.

Note that the authors of \cite{5,9} solve this eigenvalue problem in the so 
called Jackiw–-Rebbi representation of the Dirac equation corresponding to $
\mathbf{\alpha}=\mathbf{\sigma}_y$ and $\mathbf{\beta}=\mathbf{\sigma}_x$ in 
which case Eq.~\eqref{eq:14} takes the form
\begin{equation}
\int_{x_L}^{x_R}K+\frac{E\dot{V}}{2K}\left(\frac{\psi_1^2-\psi_2^2}{\dot{\psi}
_1\psi_2-\dot{\psi}_2\psi_1}\right)\rmd x=(n+1)\pi,\label{eq:23}
\end{equation}
where the overhead dot now denotes differentiation w.r.t. $x$. Substituting 
the wave functions in the l.h.s of Eq.~\eqref{eq:23}, we obtain
\begin{widetext}
\begin{equation}\label{eq:24}
I_{\alpha}(\nu)=\pi\nu-\alpha\sqrt{2\nu-\alpha^2}\nonumber-2\nu\tan^{-1}\left(
\frac{\alpha}{\sqrt{2\nu-\alpha^2}}\right)+\int_{\alpha}^{\sqrt{2\nu}}\left(
\frac{H_{\nu}^2(x)-2\nu H_{\nu-1}^2(x)}{H_{\nu}^2(x)-H_{\nu-1}(x)H_{\nu+1}(x)
}\right)\frac{\rmd x}{\sqrt{2\nu-x^2}}
\end{equation}
\end{widetext}
the derivation of which makes use of the recurrence formula $H_{\nu+1}(x)=2xH_
{\nu}(x)-2\nu H_{\nu-1}(x)$ \cite{hermit}. The quantization condition can be 
rewritten as $I_{\alpha}(\nu)=(n+1)\pi$. Note that for $E<mc^2+\min V$ there 
can be no turning points. Therefore, bound state solutions can only be 
expected for $\nu>\alpha^2/2$.

\begin{figure}
\begin{center}
\includegraphics[width=\columnwidth]{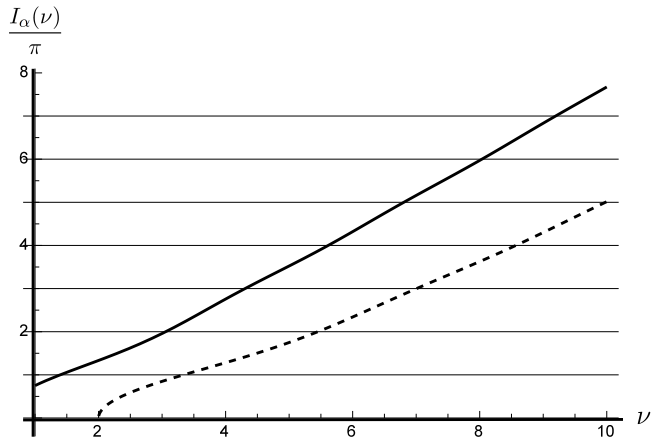}
\end{center}
\caption{Plot of $\frac{I_{\alpha}(\nu)}{\pi}$ for $\alpha=1$ (continuous line
) $\alpha=2$ (broken line).}\label{fig2}
\end{figure}

In Figure \ref{fig2} we plot the graphs of $I_{\alpha}(\nu)/\pi$ for $\alpha=1
,2$ as a function of $\nu$. These curves intersect the horizontal integer-
lines (shown in the same figure) precisely at the locations of the energy 
eigenvalues, thus confirming remarkably the validity of Eq.~\eqref{eq:14}. 
For comparison, we collect the graphical intersection points (computed using 
bisection search) in Table \ref{tab1} where the exact eigenvalues obtained 
by Cavalcanti \cite{9} are also listed.

Now, concerning potentials for which the wave functions are not obtainable in 
terms of known functions, a numerical routine can be outlined for the Dirac 
problem (following the prescription given in \cite{11} for the 
non-relativistic ATM-quantization formula) yielding eigenvalue estimates, whose 
accuracy is only appreciable for larger eigenvalues--a feature that we do 
not understand completely at present. Therefore, a discussion of this scheme 
is not given here.

\begin{table*}[t]
\centering
\caption{First five values of $\nu=E^2/2g$ for select values of $\alpha=m/
\sqrt{g}$. ($\hbar=c=1$).\label{tab1}}
\begin{ruledtabular}
\begin{tabular}{cccccc}
        & \multicolumn{2}{c}{$\alpha=1$} & & \multicolumn{2}{c}{$\alpha=2$} \\
\cline{2-3} 
\cline{5-6}
$\nu$   & Exact      & ATM    & & Exact      & ATM \\
\hline
$\nu_0$ & 1.396274\underline{44057259}\footnote{Underlined digits are not 
given in \cite{9}, and can be obtained by solving $H_{\nu}^2(\alpha)=2\nu H_{
\nu-1}^2(\alpha)$.\label{fa}} & 1.39627444809303  && 3.338595\underline{
40177509}$^{\rm a}$  & 3.33859536647797\\
$\nu_1$ & 3.056760\underline{24015993}     & 3.05676024192944  && 5.452160
\underline{76495126}      & 5.45216075601056\\
$\nu_2$ & 4.306276\underline{64769999}     & 4.30627665789798  && 7.006087
\underline{30469830}      & 7.00608729608357\\
$\nu_3$ & 5.615210\underline{82352847}     & 5.61521084997803  && 8.568945
\underline{86286508}      & 8.56894588172436\\
$\nu_4$ & 6.804771\underline{21323347}     & 6.80477123537566  && 9.978608
\underline{33615064}      & 9.97860836439766
\end{tabular}
\end{ruledtabular}
\end{table*}

Next, we briefly mention another interesting application, to motivate further 
work. In fact, a detailed account of the same will be addressed in a follow-
up paper. Particularly, we consider the problem of reconstructing the ground 
state wave function $\Psi_o$, for a given potential $V(x)$. The basic scheme 
is to assume the ground state energy $E_o$ to be a parameter and invert Eq. 
\eqref{eq:14} to obtain $\Psi_o(E_o)$. Although this is easier said than 
done, for simple power-law potentials $V(x)= g|x|^n(g>0,\ n\in\mathbb{N})$, 
using the ansatz
\begin{equation*}
P=\exp\left(\int\frac{\rmd \xi}{\sum_{j=0}^{\infty}a_j\xi^j}\right)
\end{equation*}
in Eq.~\eqref{eq:10}, one can obtain the coefficients $a_j$ by comparing 
powers of $\varepsilon_o$ on either side of the resulting equation. We would 
then resort to a variational principle to minimize the energy functional $E[
\Psi]$, which is the expectation of the Dirac Hamiltonian w.r.t. the 
obtained parametric wave function, to find $E_o$. Similarly, recovery of the 
potential with knowledge of the ground state wave function is conceivable. 
However, this inversion problem is significantly complicated by the presence 
of the effective turning points, which depend on the potential implicitly.

\section{Conclusion}\label{sec5}
In this paper we analyzed the bound states of the one-dimensional Dirac 
equation for a scalar potential well $V(x)$ (satisfying two constraints laid 
in Section~\ref{sec1}), obtaining an exact energy quantization formula that 
extends the previously known analog for the Schr\"{o}dinger equation. In 
fact, we could show that our formula reproduces this result in the 
appropriate non-relativistic limit. Further, we discussed the physically 
desirable symmetries of the quantization formula and applied the same to 
compute the energy eigenvalues for the potential $V=g|x|$.

Before concluding, we invite interested readers to strive for a completion of 
our relativistic quantization formula that would apply to an arbitrary 
scalar potential. Such a generalization must address two cases for which our 
analysis fails. First, the case of non-confining potentials: consider, for 
example, the one-dimensional Dirac equation with the scalar exponential 
potential $V(x)=Ae^{-\lambda x},\ \lambda>0$ \cite{10} which, despite being 
repulsive everywhere, supports bound states for $A<0$. This is a consequence 
of the nature of coupling in the Dirac equation. Note that the exponential 
potential yields only one effective turning point, thus failing to satisfy (
at least) one of the constraints laid in Section~\ref{sec1}. As a result, 
our analysis would not apply to this case. Second, one must account for 
potentials that yield more than two effective turning points (a double well 
potential for instance). Unfortunately, the generalization of the analytic 
transfer matrix method to such potentials even for the non-relativistic 
problem stands open.

\subsection*{Acknowledgements} 
I would like to thank Dr. Hemalatha Thiagarajan, Dr. S.D. Mahanti and Dr. 
Mike Wilkes for critically reviewing the manuscript. Garv Chauhan and 
Avinash Prabhakar helped with the numerical results in Section~\ref{sec4}. 
Thanks are also due to the anonymous reviewers for their valuable suggestions.

\appendix
\section{}\label{appA}
Let $\Psi$ be a bound state wave function, with its components $\psi_{1,2}$ 
satisfying Eqs. \eqref{eq:1} and \eqref{eq:2}. Based on the properties of an 
admissible bound state wave function we deduce an useful property of the 
auxiliary function
\begin{equation}
P(\xi)=-\frac{\psi_2}{\psi_1}\label{eq:A.1}
\end{equation}
which is well defined (and bounded) at any finite $\xi$ excepting the nodes 
of $\psi_1$. We show that
\begin{equation}
P(\xi_L)<0<P(\xi_R),\label{eqA.2}
\end{equation}
where $\xi_L(\xi_R)$ denotes the left(right) effective turning point (defined 
by $\vartheta(\xi_{L,R})+1=\varepsilon)$. Using this result we obtain the `
half phase losses' at these turning points to be $\pi/2$ (Section~\ref{sec2})
. We recall that the chosen representation of the Dirac matrices allows us 
to work with a real $\Psi$ \cite{4}. Hence, the inequality in proposition 
\eqref{eqA.2}
 is valid. In proving the above proposition, we require the following lemma.
\begin{lemma}
The components of $\Psi$ cannot vanish simultaneously at any finite $
\xi$. Equivalently, $\nexists~\xi\in\mathbb{R}\vert\Psi(\xi)=0$.
\end{lemma}

\begin{proof} 
Suppose there exists a $\xi_o$ such that $\Psi(\xi_o)=0$. 
Consider a finite interval containing $\xi_o$ in which the wave function can 
be represented as ($A$, $B$ arbitrary coefficients)
\begin{equation}
\Psi(\xi)=\begin{pmatrix}1\\\frac{~~~1}{1+\vartheta(\xi)+\varepsilon}\frac{
\rmd }{\rmd \xi}\end{pmatrix}\left(A\psi_1^{(1)}(\xi)+B\psi_1^{(2)}(\xi)
\right)
\end{equation}
where $\psi_1^{(1)},\psi_1^{(2)}$ are the linearly independent solutions of 
the second order ODE
\begin{equation}
\ddot{\psi}_1-\left(\frac{\dot{\vartheta}(\xi)}{\varepsilon+\vartheta(\xi)+1}
\right)\dot{\psi}_1+\left(\frac{\varepsilon+\vartheta(\xi)+1}{\varepsilon-
\vartheta(\xi)-1}\right)\psi_1=0\label{eqA.5}
\end{equation} that results from decoupling Eqs. \eqref{eq:1} and \eqref{eq:2}.
\footnote{Although Eq.~\eqref{eqA.5} becomes singular at the effective 
turning points, solutions $\psi_1^{(1)},\psi_1^{(2)}$ can be prescribed (in 
the vicinity of these points) in accordance with the theory of singular 
differential equations (Cf. Sec 2.7 of Ref.~[\onlinecite{12}])}

Since $\Psi$ vanishes at $\xi_o$, the matrix equation
\begin{equation}  \label{eqA.5a}
\begin{pmatrix}
{\psi_1^{(1)}(\xi_o)}&{\psi_1^{(2)}(\xi_o)}\\{\dot{\psi}_1^{(1)}(\xi_o)}
&{\dot{\psi}_1^{(2)}(\xi_o)}\end{pmatrix}\begin{pmatrix}A\\B\end{pmatrix}=0
\end{equation}
must hold. Thus, in order to prevent the wave function from vanishing 
identically (i.e. $A=B=0$) on the interval, we must have $W(\psi_1^{(1)}(\xi_
o),\psi_1^{(2)}(\xi_o))=0$ where $W$ is the Wronskian of the two functions. 
But this would contradict the linear independence of the solutions $\psi_1^{(
1)},~\psi_1^{(2)}$ at $\xi_o$, hence such a $\xi_o$ does not exist.
\end{proof}
The regions where $\vartheta(\xi)+1<(>)\varepsilon$ are designated as 
allowed (forbidden) regions. Consider the properties 
\begin{description}
  \item[P1] $\psi_{1,2}\rightarrow 0$ as $\vert\xi\vert\rightarrow\infty$
  \item[P2] $\psi_{1,2}\neq 0$ for any $\xi$ (no nodes) in the forbidden region
\end{description}
which hold for any bound state wave function. We now prove 
inequality \eqref{eqA.2}.

\begin{proof} 
Since Eqs. \eqref{eq:1} and \eqref{eq:2} remain form invariant under the 
transformation $\xi\mapsto-\xi(\Rightarrow\xi_L\mapsto\xi_R)$ and $\psi_1
\mapsto-\psi_1$ while $P\mapsto-P$, it suffices to prove any one of the two 
inequalities in \eqref{eqA.2}. We focus on the right effective turning point 
$\xi_R$. The truth of proposition \eqref{eqA.2} rejects the possibility $
\mathrm{sgn}[\psi_1(\xi_R )]=\mathrm{sgn}[\psi_2(\xi_R )]$ where $\mathrm{sgn
}[~]$ is the signum function. To prove this we let $\psi_{1,2}(\xi_R)>0$. 
From \textbf{P2}, it follows that $\psi_1>0$ for all $\xi>\xi_R$ (a 
forbidden region). Since $\psi_2(\xi_R)>0$, Eq.~\eqref{eq:1} implies that $
\dot{\psi}_1(\xi_R)>0$. Since $\psi_1$ is increasing at $\xi_R$, it must 
attain at least one maximum before vanishing asymptotically as $
\xi\rightarrow\infty$ (property \textbf{P1}), remaining positive-definite 
all along. Clearly, at the site of this maximum, $\dot{\psi}_1=0
\Rightarrow\psi_2=0$, which contradicts property \textbf{P2} for $\psi_2$. 
The other possibility $\psi_{1,2}(\xi_R)<0$ (a `reflection' of the previous 
case) is readily contradicted from form-invariance of Eqs. \eqref{eq:1} and 
\eqref{eq:2} under the transformation $\Psi\mapsto-\Psi$. Thus, $\mathrm{sgn}
[\psi_1(\xi_R )]=-\mathrm{sgn}[\psi_2(\xi_R )]\Rightarrow P(\xi_R)>0$.
\end{proof}

Finally, we show that
\begin{equation}
|P(\xi_{L,R})|<\infty\label{eqA.6}
\end{equation}
\begin{proof} 
Since $\psi_1$ and $\psi_2$ cannot vanish simultaneously (using the above 
Lemma), we need only show that an effective turning point cannot be a node 
of $\psi_1$. Consider $\xi_R$ as before. Assume $\psi_1(\xi_R)=0$. From 
\textbf{P1}, $\psi_1\rightarrow 0$ as $\xi\rightarrow\infty$. Thus, $\psi_1$ 
must attain at least one minimum (maximum) if $\psi_2(\xi_R)<0(>0)$ ($
\because \mathrm{sgn}[\dot{\psi_1}]=\mathrm{sgn}[\psi_2]$), remaining 
negative (positive) definite all along. Similar to our earlier argument, at 
the site of this extremum $\psi_2=0$, which contradicts property \textbf{P2}
. Thus, $\psi_1(\xi_R)\neq 0$. 
\end{proof}

\section{}\label{appB}
In this appendix we show that the integral (Eq.~\eqref{eq:10})
\begin{equation}
I=\int_{\xi_L}^{\xi_R}\frac{P\dot{\lambda}}{P^2+\lambda^2}\rmd \xi\label{B.1}
\end{equation}
exists for all potentials $\vartheta(\xi)$ that satisfy the constraints laid 
in Section~\ref{sec1} namely, (i) $\mathrm{min}\vartheta+1\geq 0$ and (ii) $1
+\vartheta(\xi)=\varepsilon$ holds only at the \textit{two} effective 
turning points $\xi_{L,R}$. In the analysis that follows, it is important to 
bear in mind that the components $\psi_{1,2}$ of a bound state wave function 
$\Psi$ have only finite number of zeros (or nodes) for $\xi_L<\xi<\xi_R$. 
Furthermore, between two consecutive nodes of $\psi_1$ there exists \textit{
exactly} one node of $\psi_2$ and vice-versa, which is a consequence of the 
fact that between two consecutive nodes of $\psi_1$ (say) there exists 
\textit{at least} one extreme point of $\psi_1$ at which $\psi_2\propto\dot{
\psi_1}=0$ (see Eq.~\eqref{eq:2}). Since this argument holds good for both $
\psi_1$ and $\psi_2$, there could \textit{at the most} be one node of $\psi_2
$ between two consecutive nodes of $\psi_1$. 

Now, we show that $I$ does not fail to exist due to the singularity at $\xi_R
$. To do so, we focus on an interval $[\xi_a,\xi_R)$ where $\xi_a$ is 
greater than the largest among the nodes of $\psi_1$, $\psi_2$ and the (only)
 point of minimum of $\vartheta$ (denoted by $\xi_c$). Defining
\begin{equation}
I_a\defeq\int_{\xi_a}^{\xi_R}\frac{P\dot{\lambda}}{P^2+\lambda^2}\rmd 
\xi\label{B.2}
\end{equation}
we have
\begin{equation}
|I_a|\leq\int_{\xi_a}^{\xi_R}\left|\frac{P\dot{\lambda}}{P^2+\lambda^2}\right|
\rmd \xi\leq\int_{\xi_a}^{\xi_R}\left|\frac{\dot{\lambda}}{P}\right|\rmd 
\xi\label{B.3}
\end{equation}
Using Eq.~\eqref{eq:11} we find
\begin{equation*}
|\dot{\lambda}|= \frac{\varepsilon}{1+\vartheta+\varepsilon}\frac{|\dot{
\vartheta}|}{\sqrt{\varepsilon^2-(1+\vartheta)^2}}
\end{equation*}
Also, as $1+\mathrm{min}\vartheta\geq 0$ (constraint (i)),
\begin{equation*}
\frac{\varepsilon}{1+\vartheta+\varepsilon}\leq1
\end{equation*}
combining these results, we obtain
\begin{align}
|I_a|\leq \int_{\xi_a}^{\xi_R}\frac{|\dot{\vartheta}|}{\sqrt{\varepsilon^2-(1+
\vartheta)^2}}\frac{\rmd \xi}{|P|}\nonumber\\=\int_{\xi_a}^{\xi_R}\frac{\dot{
\vartheta}}{\sqrt{\varepsilon^2-(1+\vartheta)^2}}\frac{\rmd \xi}{|P|}
\label{eqB.4}
\end{align}
since $\dot{\vartheta}(\xi)>0$ for $\xi>\xi_c$. Integrating the r.h.s. of 
\eqref{eqB.4} by parts we obtain
\begin{align*}
|I_a|&\leq \frac{\pi}{2|P(\xi_R)|}-\frac{\sin^{-1}\left(\frac{1+\vartheta(\xi_
a)}{\varepsilon}\right)}{|P(\xi_a)|}\\&-\int_{\xi_a}^{\xi_R}\sin^{-1}\left(
\frac{1+\vartheta}{\varepsilon}\right)\frac{\rmd }{\rmd \xi}\left(\frac{1}{|P
|}\right)\rmd \xi
\end{align*}
Note that $|P(\xi_a)|$ is non-zero and finite, owing to the choice of the 
point $\xi_a$ and so is $|P(\xi_R)|$ (inequalities~\eqref{eqA.2} 
and~\eqref{eqA.6}). The r.h.s above can be bounded from above by its absolute value 
leading to
\begin{equation}
|I_a|\leq \frac{\pi}{2}\left(\frac{1}{|P(\xi_R)|}+\frac{1}{|P(\xi_a)|}+\int_{
\xi_a}^{\xi_R}\left|\frac{\rmd }{\rmd \xi}\left(\frac{1}{|P|}\right)\right|
\rmd \xi\right)\label{B.5}
\end{equation}
Note that, $\mathrm{sgn}(P(\xi))=\mathrm{sgn}(P(\xi_R))$ for all $\xi_a<\xi<
\xi_R$, as $P$ doesn't cross the $\xi$ axis (no nodes of $\psi_2$). Thus, $|P
|=P$ (see inequality \eqref{eqA.2}). Furthermore, the derivative of $1/P$ 
has no zeros in this interval, since 
\begin{equation*}
\frac{\rmd }{\rmd \xi}\left(\frac{1}{P}\right)=0\Rightarrow P^2=\frac{1+
\vartheta-\varepsilon}{1+\vartheta+\varepsilon}<0
\end{equation*}
which refutes the fact that $P$ is real in the chosen representation of the 
Dirac matrices. Therefore, 
\begin{align*}
\int_{\xi_a}^{\xi_R}\left|\frac{\rmd }{\rmd \xi}\left(\frac{1}{|P|}\right)
\right|\rmd \xi&=\int_{\xi_a}^{\xi_R}\frac{\rmd }{\rmd \xi}\left(\frac{1}{P}
\right)\rmd \xi\\&=\frac{1}{P(\xi_R)}-\frac{1}{P(\xi_a)}
\end{align*}
Substituting this result in inequality \eqref{B.5} we obtain
\begin{equation}
|I_a|\leq\frac{\pi}{P(\xi_R)}
\end{equation}
which is finite (due to inequalities \eqref{eqA.2} and \eqref{eqA.6}). A 
similar argument works for the left turning point.

\end{document}